\documentclass[runningheads]{llncs}
\usepackage{amsmath,amsfonts}
\pdfoutput=1
\usepackage{xcolor}
\usepackage{nicefrac}
\usepackage{proof}
\usepackage{thmtools}
\usepackage{thm-restate}

\makeatletter
\let\MYcaption\@makecaption
\makeatother
\usepackage{subcaption}
\usepackage{graphicx}

\allowdisplaybreaks[1]

\newcommand\Sec[1] {Sect.~\ref{#1}}

\newcommand\Prop[1]{Proposition~\ref{#1}}

\newcommand\Def[1] {Definition~\ref{#1}}

\newcommand\Fig[1] {Fig.~\ref{#1}}

\newcommand{\eqdef}{\ensuremath{\stackrel{\mathrm{def}}{=}}}

\newcommand{\Dists}{\mathbb{D}} 
\newcommand{\Prob}{\mathrm{Pr}}
\renewcommand{\Pr}{\operatornamewithlimits{\mathrm{Pr}}}

\newcommand{\expect}{\operatornamewithlimits{\mathbb{E}}}

\newcommand{\reals}{\mathbb{R}}
\newcommand{\realspos}{\mathbb{R}_{>0}}
\newcommand{\realsnng}{\mathbb{R}_{\ge0}}
\newcommand{\sphere}{\mathbb{S}}

\newcommand{\calh}{\mathcal{H}}

\newcommand{\calr}{\mathcal{R}}
\newcommand{\cals}{\mathcal{S}}

\newcommand{\calu}{\mathcal{U}}
\newcommand{\calv}{\mathcal{V}}

\newcommand{\calx}{\mathcal{X}}
\newcommand{\caly}{\mathcal{Y}}

\newcommand{\bm}[1]{{\mbox{\boldmath $#1$}}}

\newcommand{\bmr}{\bm{r}}
\newcommand{\bmv}{\bm{v}}

\newcommand{\bmx}{\bm{x}}
\newcommand{\bmy}{\bm{y}}

\newcommand{\DP}[0]{\textsf{DP}}
\newcommand{\XDP}[0]{\textsf{XDP}}

\newcommand{\PXDP}[0]{\textsf{PXDP}}
\newcommand{\LDP}[0]{\textsf{LDP}}
\newcommand{\CDP}[0]{\textsf{CDP}}
\newcommand{\CXDP}[0]{\textsf{CXDP}}

\newcommand{\cosdis}[0]{\mathit{d_{\theta}}}

\newcommand{\deuc}[0]{\mathit{d}_{\rm euc}}

\newcommand{\dV}[0]{\mathit{d}_{\calv}}
\newcommand{\dX}[0]{\mathit{d}_{\calx}}

\newcommand{\distH}[0]{D_{\!\calh}}

\newcommand{\hproj}[0]{\mathit{h}_{\sf proj}}
\newcommand{\Hproj}[0]{\mathit{H}_{\sf proj}}

\newcommand{\Bino}[0]{\mathit{Binomial}}

\newcommand{\DKL}[0]{\mathit{D}_{\rm KL}}

\newcommand{\alg}{\mathit{A}}

\newcommand{\Qlap}{\mathit{Q}_{\sf Lap}}
\newcommand{\Qrr}{\mathit{Q}_{\sf rr}}
\newcommand{\Qbrr}{\mathit{Q}_{\sf brr}}

\newcommand{\Qlaplsh}{\mathit{Q}_{\sf LapLSH}}
\newcommand{\QLapH}{\mathit{Q}_{{\sf Lap}H}}
\newcommand{\Qlsh}{\mathit{Q}_{\sf LSHRR}}

\newcommand{\hash}[0]{\mathit{h^{*}}}
\newcommand{\deh}[0]{d_{\varepsilon H}}

\newcommand{\Loss}[1]{\mathcal{L}_{#1}}

\newif\ifcommentson\commentsonfalse

\newif\ifconferenceon\conferenceonfalse
\ifconferenceon
\newcommand{\arxiv}[1]{}
\newcommand{\conference}[1]{#1}
\newcommand{\conferenceShort}[1]{}
\else
\newcommand{\arxiv}[1]{#1}
\newcommand{\conference}[1]{}
\newcommand{\conferenceShort}[1]{}
\fi
\newcommand{\journalversion}[1]{}

\newcommand{\ourtablesize}[0]{.75\textwidth}

\newcommand{\myqed}[0]{\qed}

\begin{document}
\title{Locality Sensitive Hashing with Extended Differential Privacy 
\thanks{The authors are ordered alphabetically.
This work was supported by the French-Japanese project LOGIS within the Inria Equipes Associ\'{e}es program,
by an Australian Government RTP Scholarship (2017278),
by ERATO HASUO Metamathematics for Systems Design Project (No. JPMJER1603), JST, and 
by JSPS KAKENHI Grant Number JP19H04113.
}}
\titlerunning{LSH with Extended Differential Privacy}
\author{
Natasha Fernandes\inst{1,2}
\and
Yusuke Kawamoto\inst{3}
\and
Takao Murakami\inst{3}
}
\institute{
Macquarie University, Sydney, Australia
\and
Inria, \'{E}cole Polytechnique, IPP, Palaiseau, France
\and
National Institute of Advanced Industrial Science \\ 
and Technology (AIST), Tokyo, Japan \\
}
\maketitle              
\begin{abstract}
Extended differential privacy, a generalization of standard differential privacy (DP) using a general metric, has been widely studied to provide rigorous privacy guarantees while keeping high utility. 
However, existing works on extended DP are limited to few metrics, such as the Euclidean metric.
Consequently, they have only a small number of applications, such as location-based services and document processing. 

In this paper, we propose a couple of mechanisms providing extended DP with a different metric: \emph{angular distance} (or \emph{cosine distance}). 
Our mechanisms are based on locality sensitive hashing (LSH), which can be applied to the angular distance and work well for personal data in a high-dimensional space. 
We theoretically analyze the privacy properties of our mechanisms, and prove extended DP for input data by taking into account that LSH preserves the original metric only approximately. 
We apply our mechanisms to friend matching based on high-dimensional personal data with angular distance in the local model, and evaluate our mechanisms using two real datasets. 
We show that LDP requires a very large privacy budget and that RAPPOR does not work in this application.
Then we show that our mechanisms enable friend matching with high utility and rigorous privacy guarantees based on extended DP.

\end{abstract}

\section{Introduction}
\label{sec:intro}
\emph{Extended differential privacy} (extended \DP{}), a.k.a. \emph{$\dX$-privacy} \cite{Chatzikokolakis:13:PETS}, is 
a privacy notion that provides rigorous privacy guarantees while enabling high utility. 
Extended \DP{} is a generalization of standard \DP{} \cite{Dwork:06:ICALP,Dwork:06:TCC} in that the adjacency relation (regarded as the Hamming distance) is generalized to a metric. 
A well-known application is geo-indistinguishability 
\cite{Alvim:18:CSF,Andres:13:CCS,Bordenabe:14:CCS}, 
an instance of extended \DP{} for two-dimensional Euclidean space. 
Geo-indistinguishability guarantees that a user's location is indistinguishable from any location within a certain radius (e.g., 
within 5km) in the local model, in which each user obfuscates her own data and sends it to a data collector. 
It can also be regarded as a \emph{relaxation} of \DP{} in the local model (local \DP{} or \LDP{} \cite{Duchi:13:FOCS}) to make two locations within a certain radius indistinguishable (whereas \LDP{} makes arbitrary locations indistinguishable).
Consequently, 
extended \DP{}
results in 
much higher utility than \LDP{}, 
e.g., for a task of estimating 
geographic population distributions \cite{Alvim:18:CSF}.

Since extended \DP{} is defined using a general metric, it can potentially have a wide range of applications. 
However, the range of actual applications is limited by the particular metrics for which extended \DP{} mechanisms have been designed.
For example, the existing works on locations 
\cite{Alvim:18:CSF,Andres:13:CCS,Bordenabe:14:CCS}, 
documents \cite{Fernandes19:POST}, 
range 
queries \cite{Xiang_ISIT20}, 
and linear queries 
\cite{Kamalaruban_PoPETs20} 
are designed for
the Euclidean metric, 
the Earth Mover's metric, 
the $l_1$ metric, 
and the summation of privacy budgets for attributes, 
respectively. 
However, there have been no known extended \DP{} mechanisms designed for the angular distance (or cosine distance).

For example, consider friend matching (or friend recommendation) based on personal data (e.g., locations, rating history) 
\cite{Brendel_WWW18,chen2015preserving,cheng2018efficient,Li_IoTJ17,liu2016linkmirage,ma2018armor,narayanan2011location,samanthula2015privacy}. 
In the case of locations, 
we can create a vector of visit-counts where each value is the visit-count on the corresponding Point of Interest (POI). 
Users with similar vectors have a high probability of establishing new friendships~\cite{Yang_WWW19}. 
Therefore, we can use the POI vector to recommend a new friend. 
Similarly, 
we 
can 
recommend a new friend based on the similarity of their item rating vectors,
since this identifies users with similar interests \cite{RecommenderSystems_book}. 
Because the distance between vectors in such applications is usually given by the angular distance (or equivalently, the cosine distance) \cite{RecommenderSystems_book}, 
the angular distance is a natural choice 
for the utility measure and the metric for extended \DP{}.

In this paper, we focus on friend matching in the local model, and 
propose 
two mechanisms
providing extended \DP{} 
with the angular distance. 
Our mechanisms are 
based on \emph{locality sensitive hashing} (LSH) \cite{Gionis_VLDB99,Wang_ProcIEEE16}, 
which 
can be applied 
to a wide range of metrics including the angular distance. 
Our first mechanism, LapLSH, uses the 
multivariate Laplace mechanism \cite{Fernandes19:POST} to generate noisy vectors, 
and then hashes them into buckets using LSH as post-processing. 
Our second mechanism, LSHRR, 
embeds personal data 
into a binary vector 
using LSH, and then applies Warner's randomized response \cite{Warner_JASA65} for each bit of the binary vector. 

The privacy analysis of extended \DP{} is challenging especially for LSHRR. 
This is because LSH does not precisely preserve 
the original metric; 
it \emph{approximates} the original 
metric 
via hashing. 
We theoretically analyze the privacy properties of our mechanisms, 
showing
that they provide extended \DP{} for the input.
We also note that 
much existing work on privacy-preserving LSH \cite{aghasaryan2013use,Chen_WCN19,qi2017distributed} 
fails to provide rigorous guarantees about user privacy. 
We point out, using a toy example, how the lack of rigorous guarantees can lead to privacy 
breaches.

We 
evaluate our mechanisms 
using two real datasets. 
We show that \LDP{} requires a very large privacy budget $\epsilon$. 
This comes from the fact that \LDP{} expresses an upper bound on the privacy guarantee for all inputs.
In contrast, extended \DP{} is a finer-grained notion than \LDP{} in that it describes the privacy guarantee for inputs at various distances. 
In fact, we show that extended \DP{} 
enables 
friend matching with a much smaller privacy budget than \LDP{} for close inputs.

We also 
explain 
why 
RAPPOR \cite{Erlingsson_CCS14} and the generalized RAPPOR \cite{Wang_USENIX17}, which are state-of-the-art 
\LDP{} 
mechanisms, 
cannot be applied
(either completely lose utility or are computationally infeasible) to friend matching. 
In short, the Bloom filter used in RAPPOR is not a metric-preserving hashing, and therefore cannot guarantee utility w.r.t. the metric distance between user vectors. This is further elaborated in \Sec{sub:RAPPOR}.

\smallskip
\noindent{\textbf{Contributions.}}
Our 
main contributions 
are as follows:
\begin{itemize}
\item We propose 
two 
mechanisms 
providing extended \DP{} with the angular distance: 
LapLSH and LSHRR. 
We show that LSH itself does not provide privacy guarantees and could result in complete privacy collapse in some situations. 
We then prove that our 
mechanisms 
provide 
rigorous guarantees of extended \DP{}.
In particular, we show that the distribution of the LSHRR's privacy loss can be characterized as extended notions of concentrated \DP{}~\cite{Dwork:16:arXiv} and probabilistic \DP{}~\cite{Machanavajjhala:08:ICDE} with input distance. 
To our knowledge, this work is the first to provide extended \DP{} with the angular distance.
\item We apply our 
mechanisms 
to friend matching based on rating history and locations. 
Then we compare 
LSHRR with LapLSH using 
two real datasets. 
We 
show 
that 
LSHRR provides 
higher (resp.~lower) utility 
than LapLSH for a high-dimensional (resp.~low-dimensional) vector. 
We also show that \LDP{} requires a very large privacy budget $\epsilon$, and RAPPOR does not work for friend matching. 
Finally, we 
show that 
LSHRR provides 
high utility for a high-dimensional vector (e.g., $1000$-dimensional rating/location vector) 
in the medium privacy regime \cite{Acharya_AISTATS19,Ye_ISIT17} of extended \DP{}, and therefore 
enables 
friend matching with rigorous privacy guarantees and high utility.
\end{itemize}

\conference{All proofs on the technical results can be found in the preprint~\cite{Fernandes_arXiv20}.}
\arxiv{All proofs on the technical results can be found in Appendix~\ref{sec:proofs}.}

\section{Related Work}
\label{sec:related}

\subsection{Extended \DP{}}
\label{sub:related:DP_XDP}

As explained in \Sec{sec:intro}, 
there are a number of 
existing extended \DP{} mechanisms 
\cite{Alvim:18:CSF,Andres:13:CCS,Bordenabe:14:CCS,Fernandes19:POST,Kamalaruban_PoPETs20,Xiang_ISIT20} 
designed for other metrics
(e.g., the Euclidean metric, the $l_1$ metric), 
which
cannot be applied to 
the angular distance. 
To our knowledge, our mechanisms are the first to provide extended \DP{} with the angular distance.

In addition, most of the studies on extended \DP{} have 
studied 
low-dimensional 
data 
such as two-dimensional 
\cite{Alvim:18:CSF,Andres:13:CCS,Bordenabe:14:CCS,Kamalaruban_PoPETs20} 
and six-dimensional \cite{Xiang_ISIT20} data.
One exception is the work in \cite{Fernandes19:POST}, 
which proposed the multivariate Laplace mechanism for 
$300$-dimensional 
vectors. 
In this paper, 
we 
apply our mechanisms to 
vectors in $1000$-dimensions 
(much larger than any existing work), and show that our LSHRR provides high utility for such high-dimensional data.

\subsection{Privacy-Preserving Friend Matching}
\label{sub:related:ppfm}
A number of studies 
\cite{Brendel_WWW18,chen2015preserving,cheng2018efficient,Li_IoTJ17,liu2016linkmirage,ma2018armor,narayanan2011location,samanthula2015privacy} 
have been made on algorithms for privacy-preserving friend matching (or friend recommendation). 
Many of them 
(e.g., \cite{cheng2018efficient,ma2018armor,narayanan2011location,samanthula2015privacy}) 
use cryptographic techniques such as homomorphic encryption and secure multiparty computation. 
However, such techniques require high computational costs or focus on specific algorithms, and are not suitable for a more complicated calculation of distance such as the angular distance between two rating/location vectors. 

The techniques in \cite{Brendel_WWW18,chen2015preserving,Li_IoTJ17,liu2016linkmirage}
are based on perturbation. 
The mechanisms in \cite{Brendel_WWW18,Li_IoTJ17,liu2016linkmirage} do not provide
\DP{} or its variant, whereas that in \cite{chen2015preserving} provides \DP{}. 
The technique in \cite{chen2015preserving}, however, is based on social graphs and cannot be applied to our setting, where a user's personal data is represented as a rating vector or visit-count vector. 
Moreover, \DP{}-based friend matching in social graphs can require prohibitive trade-offs between utility and privacy \cite{Brendel_WWW18,Machanavajjhala_PVLDB11}. 

Similarly, \DP{} mechanisms based on each user's high-dimensional
rating/location 
vector require a very large privacy budget (e.g., $\varepsilon \geq 250$ \cite{Liu_RecSys15}, $\varepsilon \geq 2 \times 10^4$ \cite{Murakami_PoPETs21}) to provide high utility.
In contrast,
our extended \DP{} mechanisms provide 
meaningful privacy guarantees in high-dimensional spaces with 
high utility,
since 
extended \DP{} is a finer-grained notion than \DP{}, 
as explained in \Sec{sec:intro}.

We also note that a privacy-preserving clustering algorithm in \cite{Nissim_PMLR18} and an item recommendation algorithm in \cite{Shin_TKDE18} cannot be applied to friend matching. 

\subsection{Privacy-Preserving LSH}
\label{sub:related:PPLSH}
Finally, we note that some studies have proposed privacy-preserving LSH 
\cite{aghasaryan2013use,Aumuller_arXiv20,Chen_WCN19,Hu_CCPE21,Nissim_PMLR18,qi2017distributed,Zhang_ICONIP20}. 
However, 
some of them \cite{aghasaryan2013use,Chen_WCN19,qi2017distributed} only apply LSH and claim that it protects user privacy because LSH is a kind of non-invertible transformation. 
In 
\Sec{sec:prv_prop}, 
we
show
that the lack of rigorous guarantees can lead to privacy breaches.
Nissim and Stemmer \cite{Nissim_PMLR18} proposed clustering algorithms based on LSH and 
the
heavy-hitters algorithm. 
However, their algorithms focus on clustering such as $k$-means clustering and cannot be applied to friend matching.

Aum\"{u}ller \textit{et al.} \cite{Aumuller_arXiv20} proposed 
a privacy-preserving LSH algorithm that can be applied to friend matching. 
Specifically, 
they 
focused on 
a similarity search problem under the Jaccard similarity using up to 2000-dimensional vectors, and proposed an \LDP{} algorithm based on MinHash. 
After the submission of our paper to a preprint \cite{Fernandes_arXiv20}, 
two related papers \cite{Hu_CCPE21,Zhang_ICONIP20} have been published. 
Zhang \textit{et al.} \cite{Zhang_ICONIP20} proposed 
an 
\LDP{} 
algorithm for rating prediction based on MinHash and knowledge distillation. 
Hu \textit{et al.} \cite{Hu_CCPE21} proposed 
an \LDP{} algorithm based on LSH 
for federated recommender system.

Our work 
differs from \cite{Aumuller_arXiv20,Hu_CCPE21,Zhang_ICONIP20} in the following 
points. 
First, \cite{Aumuller_arXiv20,Hu_CCPE21,Zhang_ICONIP20} only analyzed \LDP{} for hashes, and did not conduct a more challenging analysis of extended \DP{} for inputs. 
In contrast, our work provides a careful analysis of extended \DP{}, 
given that LSH preserves the original metric only approximately. 
We also show that extended \DP{} requires a much smaller privacy budget than \LDP{}.
Second, 
we compared LSHRR with LapLSH in detail, and show that LSHRR (resp.~LapLSH) is more suitable for high (resp. low) dimensional data.

\section{Preliminaries}
\label{sec:preliminaries}
In this section, we introduce notations 
and recall background on locality sensitive hashing (LSH), privacy measures, and privacy protection mechanisms.

Let $\deuc$ be the Euclidean distance between real vectors, i.e.,\, $\deuc(\bmx, \bmx') = \|\bmx - \bmx'\|_2$.
We write $\calv$ for the set of all binary data of length $\kappa$, i.e., $\calv = \{0, 1\}^\kappa$.
The \emph{Hamming distance} between $\bmv, \bmv' \in \calv$ is: 
$\dV(\bmv, \bmv') =
\sum_{i=1}^{\kappa} \left|\, v_i - v'_i \,\right|$
{.}

We denote 
the \emph{set of all probability distributions} over a set~$\cals$ by $\Dists\cals$.
Let $N(\mu, \sigma^2)$ be the normal distribution with mean $\mu$ and variance $\sigma^2$.
Let $A: \calx\rightarrow\Dists\caly$ be a randomized algorithm from a finite set $\calx$ to another $\caly$, and $A(x)[y]$ (resp. by $A(x)[S]$) be the probability that $A$ maps $x$ to $y$ (resp. an element of~$S$).

\subsection{Locality Sensitive Hashing (LSH)}
\label{sub:LSH}

We denote by $\calx$ the set of all possible input data.
We introduce 
the notion of a
\emph{(normalized) dissimilarity function} $\dX: \calx\times\calx\rightarrow[0,1]$ over $\calx$ such that
two inputs $\bmx$ and $\bmx'$ have less dissimilarity $\dX(\bmx, \bmx')$ when they are closer, and that $\dX(\bmx, \bmx') = 0$ when $\bmx = \bmx'$. %
If $\dX$ is symmetric and subadditive, it is a metric.

A \emph{locality sensitive hashing (LSH)}~\cite{Gionis_VLDB99} is a family of functions 
in which
the probability of two inputs $\bmx, \bmx'$ having different $1$-bit outputs is proportional to 
$\dX(\bmx, \bmx')$.

\begin{definition}[Locality sensitive hashing]\label{def:LSH}\rm
A \emph{locality sensitive hashing (LSH) scheme} w.r.t. a dissimilarity function $\dX$ is a family $\calh$ of functions from $\calx$ to $\{0,1\}$ coupled with a probability distribution $D_\calh$ such that 
for any $\bmx, \bmx'\in\calx$,
\begin{align}
\Pr_{h\sim \distH}\![ h(\bmx) \neq h(\bmx') ] = \dX(\bmx, \bmx')
{,}
\end{align}
where $h$ is chosen from $\calh$ according to the distribution $\distH$.
By using independently chosen functions $h_1, h_2, \ldots, h_\kappa$,
the \emph{$\kappa$-bit LSH function} $H: \calx \rightarrow\calv$ is:
\begin{align}
H(\bmx) = (h_1(\bmx),\, h_2(\bmx),\, \ldots,\, h_\kappa(\bmx))
{.}
\label{eq:kappa-bit-hash}
\end{align}
We denote by $H^*: \calx \rightarrow\Dists\calv$ the randomized algorithm that draws a $\kappa$-bit LSH $H$ from the distribution $\distH^\kappa$ and outputs the hash value $H(\bmx)$ of a given input~$\bmx$.
\end{definition}

\subsection{Examples of LSHs}
\label{sub:LSH:examples}

There are a variety of LSH families corresponding to useful metrics, such as the angular distance~\cite{Andoni_NIPS15,Charikar_STOC02}, Jaccard metric \cite{Broder_JCSS00}, and $l_p$ metric with $p \in (0,2]$ \cite{Datar_SCG04}.
In this work, we focus on LSH families for the angular distance.

A \emph{random-projection-based hashing} is a one-bit hashing 
with the domain $\calx {\eqdef} \reals^n$
and a random vector $\bmr\in\reals^n$ that defines a hyperplane through the origin.
Formally, 
we define a \emph{random-projection-based hashing} $\hproj: \reals^n {\rightarrow} \{0,1\}$~by:
\begin{align*}
\hproj(\bmx) =
\begin{cases}
    0 & (\text{if } \bmr^\top \bmx < 0) \\[-0.5ex]
    1 & (\text{otherwise})
\end{cases}
\end{align*}
where 
each element of $\bmr$ is independently chosen from the standard normal distribution $N(0, 1)$.
By \eqref{eq:kappa-bit-hash}, a $\kappa$-bit LSH function $\Hproj$ is built from one-bit %
hashes %
${\hproj}_1, \ldots, {\hproj}_\kappa$ that are generated from independent hyperplanes $\bmr_1, \ldots, \bmr_\kappa$.

The random-projection-based hashing $\hproj$ is an LSH w.r.t. 
the \emph{angular distance} 
$d_\theta: \reals^n\times\reals^n\rightarrow[0,1]$ 
defined by:
\begin{align}
\cosdis(\bmx, \bmx') &= {\textstyle \frac{1}{\pi} \cos^{-1}\big(\frac{\bmx^\top \bmx'}{\|\bmx\| \|\bmx'\|}\big)}
\label{def:cosine_distance}
\end{align}
For example, $d_\theta(\bmx, \bmx') = 0$ iff $\bmx = \bmx'$, while $d_\theta(\bmx, \bmx') = 1$ iff $\bmx = -\bmx'$.
$d_\theta(\bmx, \bmx') = 0.5$
exactly when
the two vectors $\bmx$ and $\bmx'$ are orthogonal, namely, $\bmx^\top \bmx' = 0$.

\subsection{Approximate Nearest Neighbor Search}
\label{sub:NN}

We recall the nearest neighbor search (NNS) problem and 
its utility measures.

Given a dataset $S \subseteq \calx$, the \emph{nearest neighbor search (NNS)} for an $x_0\in S$ is the problem of finding the closest $x \in S$ to $x_0$ w.r.t. a metric $\dX$ over $\calx$.
A \emph{$k$-nearest neighbor search ($k$-NNS)} 
is the problem of finding the $k$ closest points.

A naive and exact approach to $k$-NNS is to perform pairwise comparisons of data points, requiring $O(|S|)$ operations.
Approaches to improve this computational inefficiency shift the problem to space inefficiency~\cite{Andoni:18:ICM}. 
An alternative approach%
~\cite{indyk1998approximate} is to employ LSH to perform \emph{approximate} NNS efficiently.
To evaluate the utility, we use 
the average distance of returned nearest neighbors from the data point $x_0$ compared with the average distance of true nearest neighbors. 

\begin{definition}[Utility loss]\label{def:utility-loss}\rm
Let $A$ be an approximate algorithm that produces approximate $k$ nearest neighbors $N \subseteq S$ for a data point $x_0\in S$ in terms of a metric $\dX$.
The \emph{average utility loss} for $N$ w.r.t. the true nearest neighbors $T$ is given by:
$\calu_A(S) ~=~ \nicefrac{1}{k} \sum\limits_{x \in N} \dX(x_0, x) ~-~ \nicefrac{1}{k} \sum\limits_{x \in T} \dX(x_0, x)$.
\end{definition}

\subsection{Privacy Measures and Privacy Mechanisms}
\label{sub:privacy}

\emph{Extended differential privacy}~\cite{Chatzikokolakis:13:PETS,Kawamoto:19:ESORICS} 
guarantees 
that when two inputs $x$ and $x'$ are closer, their corresponding output distributions are less distinguishable.
In this paper, we propose a more generalized definition using a function $\delta$ over $\calx$ and an arbitrary function $\xi$ over $\calx$ instead of a metric.
The main reason for this generalization is that LSH preserves the metric over the input only probabilistically and approximately, hence cannot fit to  \cite{Chatzikokolakis:13:PETS}'s standard definition.

\begin{definition}[Extended differential privacy]\label{def:XDP-simple}\rm
Given two functions $\xi: \calx\times\calx\rightarrow\realsnng$ and $\delta: \calx\times\calx\rightarrow[0,1]$, a randomized algorithm $\alg: \calx \rightarrow \Dists\caly$ provides \emph{$(\xi,\delta)$-extended differential privacy (\XDP{})} if for all $x, x'\in\calx$ and for any $S\subseteq\caly$,
\begin{align*}
\alg(x)[S] \leq e^{\xi(x,x')} \,\alg(x')[S] + \delta(x,x')
{,}
\end{align*}
where the probability is taken over the random choices in $\alg$.
\end{definition}
We abuse notation and write $\delta$ when $\delta(x, x')$ is a constant.
When $\xi(x,x')$ is also a constant $\varepsilon$, the definition gives the (standard) \emph{differential privacy} (\DP{}).
When $\xi(x,x')=\dX(x, x')$ and $\delta(x, x')=0$, the definition gives $\dX$-privacy in \cite{Chatzikokolakis:13:PETS}.
In later sections, we instantiate the metric $\dX$ with the angular distance $d_\theta$.

Finally, we recall 
some
popular privacy protection mechanisms.

\begin{definition}[Laplace mechanism~\cite{Dwork:06:TCC}]\label{def:Laplace-mech}\rm
For an $\varepsilon \in\realspos$ and a metric $\dX$ over $\calx\cup\caly$, the \emph{$(\varepsilon,\dX)$-Laplace mechanism} is the randomized algorithm $\Qlap: \calx\rightarrow\Dists\caly$ that maps an input $x$ to an output $y$ with probability
${\textstyle\frac{1}{c}} \exp(-\varepsilon \dX(x,y))$
where $c = \int_{\caly} \exp(-\varepsilon \dX(x,y))~dy $.
\end{definition}
Examples of the $(\varepsilon,\dX)$-Laplace mechanism include the one-dimensional \cite{Dwork:06:TCC} and the multivariate Laplace mechanism \cite{Fernandes19:POST}, both equipped with the Euclidean metric.
The $(\varepsilon,\dX)$-Laplace mechanism provides $(\varepsilon\dX,0)$-\XDP{}.

\begin{definition}[Randomized response~\cite{Warner_JASA65}]\label{def:RR}\rm
The \emph{$\varepsilon$-randomized response} (\emph{$\varepsilon$-RR}) 
is the randomized algorithm $\Qrr: \{0,1\}\rightarrow\Dists\{0,1\}$ that maps a bit $b$ to another $b'$ 
with probability
$\frac{e^{\varepsilon}}{e^{\varepsilon}+1}$ if $b' = b$,
and with probability
$\frac{1}{e^{\varepsilon}+1}$ otherwise.
\end{definition}
The $\varepsilon$-RR provides $\varepsilon$-\DP{}.
Erlingsson et al.~\cite{Erlingsson_CCS14} introduce 
the \emph{RAPPOR}, 
which first uses a Bloom filter to produce a hash value and then applies the 
RR 
to each bit of the hash value. 
The RAPPOR provides $\varepsilon$-\DP{} 
in the local model.

\section{Privacy Properties of LSH}
\label{sec:prv_prop}
Several works in the literature make reference to the privacy-preserving properties of LSH~\cite{aghasaryan2013use,chow2012practical,qi2017distributed}. The privacy guarantee attributed to LSH mechanisms hinges on its hash function, which `protects' an individual's private attributes by revealing only their hash bucket. We now apply a formal analysis to LSH and explain why LSH implementations do not provide strong privacy guarantees, and could, in some situations, result in complete privacy collapse for the individual.

\smallskip
\noindent{\textbf{Modeling LSH.}}
We present a simple example to show how privacy breaks down. 
Consider the set of secret inputs $\calx {=} \{ (0,1), (1,0), (1,1) \}$ 
whose element represents whether an individual rated two movies $A$ and $B$. 
Then an LSH is modeled as a probabilistic channel $\hash: \calx {\rightarrow} \mathbb{D}\{0, 1\}$ that maps a secret input to a binary observation.

For brevity, we deal with a single random-projection-based hashing $h$ in \Sec{sub:LSH:examples}.
That is, we randomly choose a vector $\bmr$ representing the normal to a hyperplane,
and given an input $\bmx \in\calx$, 
the hash function $h$ outputs $0$ 
if $\bmr^\top \bmx < 0$
and $1$ otherwise.
For example, if $\bmr = (1, -{\textstyle\frac{1}{2}})$ is chosen, $h$ is defined as:
\begin{small}
\begin{align*}
         h: \calx &\rightarrow \{0,1\} \\[-0.3ex]
             (0,1) &\mapsto 0 \\[-0.5ex]
             (1,0) &\mapsto 1 \\[-0.5ex]
             (1,1) &\mapsto 1
\end{align*}
\end{small}
In fact, there are 6 possible (deterministic) hash functions for any choice of the vector $\bmr$, corresponding to hyperplanes that separate different pairs of points:
\begin{center}\begin{small}
  \begin{tabular*}{\ourtablesize}{@{\extracolsep{\fill}}c c c}
     $h_1$                    & $h_2$                   & $h_3$                   \\[-0.4ex]
     $(0,1) \mapsto 1$  & $(0,1) \mapsto 0$ & $(0,1) \mapsto 1$ \\[-0.4ex]
     $(1,0) \mapsto 0$  & $(1,0) \mapsto 1$ & $(1,0) \mapsto 0$ \\[-0.4ex]
     $(1,1) \mapsto 0$  & $(1,1) \mapsto 1$ & $(1,1) \mapsto 1$ \\[-0.4ex]
     \\[-1.9ex]
     $h_4$                   & $h_5$                   & $h_6$                   \\[-0.4ex]
     $(0,1) \mapsto 0$ & $(0,1) \mapsto 1$ & $(0,1) \mapsto 0$ \\[-0.4ex]
     $(1,0) \mapsto 1$ & $(1,0) \mapsto 1$ & $(1,0) \mapsto 0$ \\[-0.4ex]
     $(1,1) \mapsto 0$ & $(1,1) \mapsto 1$ & $(1,1) \mapsto 0$ \\
  \end{tabular*}\end{small}
\end{center}
Each of $h_1$, $h_2$, $h_3$, and $h_4$ occurs with probability $\nicefrac{1}{8}$, while $h_5$ and $h_6$ each occur with probability $\nicefrac{1}{4}$.
The resulting channel $\hash$, computed as the probabilistic sum of these deterministic hash functions, turns out to leak no information on the secret input (i.e., all outputs have equal probability conditioned on each input).

This indicates that the channel $\hash$ is perfectly private.
However, in practice, LSH may require the release of the choice of the vector $\bmr$ (e.g. \cite{chow2012practical})\footnote{In fact, since the channel on its own leaks nothing, there \emph{must} be further information released in order to learn anything useful from this channel.}, 
that is, the choice of hash function is leaked. 
Notice that in our example, $h_1$ to $h_4$ correspond to deterministic mechanisms which leak exactly 1 bit of the secret, while $h_5$ and $h_6$ leak nothing. In other words, with $50\%$ probability, 1 bit of the 2-bit secret is leaked.
Furthermore, $h_1$ and $h_2$ leak the secret $(0,1)$ exactly, and $h_3$ and $h_4$ leak $(1,0)$ exactly. 
Thus, the release of $\bmr$ destroys the privacy guarantee.

\smallskip
\noindent{\textbf{The Guarantee of LSH.}}
In general, for any number of hash functions and any length of input, an LSH which releases its choice of hyperplanes also leaks its choice of deterministic mechanism.
This means that it leaks the equivalence classes of the secrets. Such mechanisms belong to the `$k$-anonymity'-style of privacy mechanisms which promise privacy by hiding secrets in equivalence classes of size at least $k$. These have been shown to be unsafe due to their failure to compose well~\cite{Ganta:08:KDD,fernandes2018processing,Kawamoto:18:ISITA}.
This failure leads to the potential for linkage or intersection attacks by an adversary armed with auxiliary information. For this reason, we consider compositionality an essential property for a privacy-preserving system. LSH with hyperplane release does not provide such privacy guarantees.

\section{LSH-based Privacy Mechanisms}
\label{sec:LSHPM}
In this section, we propose two privacy protection mechanisms called \emph{LSHRR} and \emph{LapLSH}.
The former is an extension of RAPPOR~\cite{Erlingsson_CCS14} w.r.t. LSH,
and the latter is constructed using the Laplace mechanism and LSH.

\smallskip
\noindent{\textbf{Construction of LSHRR.}}
We introduce the \emph{LSH-then-RR privacy mechanism} (\emph{LSHRR}) as the randomized algorithm that (i) randomly chooses a $\kappa$-bit LSH function $H$, (ii) computes the $\kappa$-bit hash value $H(\bmx)$ of a given input $\bmx$, and 
(iii) applies the randomized response to each bit of $H(\bmx)$.

To formalize this, we define the \emph{$(\varepsilon, \kappa)$-bitwise RR} $\Qbrr$, which applies the randomized response $\Qrr$ to each bit of the input independently.
Formally, $\Qbrr: \calv\rightarrow\Dists\calv$ maps  a bitstring $\bmv = (v_1, v_2, \ldots, v_\kappa)$ to another $\bmy = (y_1, y_2, \ldots, y_\kappa)$ with probability
$\Qbrr(\bmv)[\bmy] {=} 
\prod_{i = 1}^{\kappa} \Qrr(v_i)[y_i]$.
Then 
LSHRR is defined as follows.

\begin{definition}[LSHRR]\label{def:LSHRR}\rm
The \emph{$\varepsilon$-LSH-then-RR privacy mechanism} (\emph{LSHRR})
instantiated with a $\kappa$-bit LSH function $H: \calx\rightarrow\calv$ is the randomized algorithm $Q_H: \calx\rightarrow\Dists\calv$ 
defined by $Q_H = \Qbrr \circ H$.
Given a distribution $D_\calh^\kappa$ of the $\kappa$-bit LSH functions,
the $\varepsilon$-LSHRR w.r.t. $D_\calh^\kappa$ is 
defined by $\Qlsh = \Qbrr \circ H^*$.
\end{definition}

LSHRR deals with two kinds of randomness:
(a) the randomness in choosing a (deterministic) LSH function $H$ from $D_\calh^\kappa$ (e.g., the random seed $\bmr$ in the random-projection-based hashing $\hproj$),
and
(b) the random noise added by the bitwise RR $\Qbrr$.
We can assume that each user of this privacy mechanism selects an input $\bmx$ independently of both kinds of randomness, since they wish to protect their own privacy when publishing~$\bmx$.

In practical settings, the same LSH function $H$ is often 
used to produce hash values of different inputs; namely, multiple hash values are dependent on an identical hash seed 
(e.g., a service provider
would generate
a hash seed so that multiple users can share the same $H$ to compare their hash values).
Furthermore, the adversary 
might 
obtain the LSH function $H$ (or the seed $\bmr$ used to produce $H$), and might learn a set of possible inputs that produce the same hash value $H(x)$ without knowing the actual input~$x$.
Therefore, the hash value $H(x)$ 
might 
reveal partial information on the input $x$ (see \Sec{sec:prv_prop}), and the bitwise RR $\Qbrr$ is crucial in guaranteeing privacy (see \Sec{sec:privacy-guarantees} for our privacy analyses).

On the other hand, $\Qbrr$ causes errors in the Hamming distance as follows:

\begin{restatable}[Error bound]{prop}{utilityQH}
\label{prop:utilityQH}
For any $x, x'\in\calx$, the expected error in the Hamming distance satisfies 
$\expect[ | \dV(Q_H(x), Q_H(x')) - \dV(H(x), \allowbreak H(x')) | ] \allowbreak \le
 \frac{2\kappa}{1 + e^{\varepsilon}}$
where the expectation is taken over the randomness in the bitwise RR.
\end{restatable}

\smallskip
\noindent{\textbf{Construction of LapLSH.}}
We also propose the \emph{Laplace-then-LSH privacy mechanism} (LapLSH) as the randomized algorithm that 
(i) randomly chooses a $\kappa$-bit LSH function $H$, 
(ii) 
applies 
the multivariate Laplace mechanism $\Qlap$ to $\bmx$, 
and 
(iii) 
computes the $\kappa$-bit hash value $H(\Qlap(\bmx))$.

\begin{definition}[LapLSH]\label{def:LapLSH}\rm
The \emph{$(\varepsilon,\dX)$-Laplace-then-LSH privacy mechanism} (\emph{LapLSH})
with a $\kappa$-bit LSH function $H: \calx\rightarrow\calv$ is the randomized algorithm $\QLapH: \calx\rightarrow\Dists\calv$ 
defined by $\QLapH = H \circ \Qlap$.
The $(\varepsilon,\dX)$-LapLSH w.r.t. a distribution $D_\calh^\kappa$ of the $\kappa$-bit LSH functions is
defined by $\Qlaplsh = H^* \circ \Qlap$.
\end{definition}

LapLSH also deals with the two kinds of randomness discussed above, and the Laplace mechanism $\Qlap$ is crucial in guaranteeing privacy.
One of the main differences from LSHRR is that LapLSH adds noise directly to the input before applying LSH whereas LSHRR adds noise after applying LSH to the input.

In \Sec{sec:evaluate} we implement the multivariate Laplace mechanism with the input domain $\calx=\reals^n$ and Euclidean distance $\deuc$
described in 
\cite{Fernandes19:POST}; 
namely, we generate additive noise by constructing a unit vector uniformly at random over the $n$-dimensional unit sphere $\sphere^{n}$, scaled by a random value generated from the gamma distribution with shape $n$ and scale $\nicefrac{1}{\varepsilon}$.

\section{Privacy Analyses of the Mechanisms}
\label{sec:privacy-guarantees}
We provide an analysis of the privacy guarantees 
provided by our mechanisms in two operational scenarios:
(i) w.r.t. an already-chosen LSH function (e.g., where it has been generated by a service provider), and
(ii) w.r.t. all possible choices of the LSH function (e.g., prior to its instantiation by a particular service provider).
Note that our analysis is general in that it does not
rely on specific metrics or hashing algorithms for LSH.

\subsection{LSHRR's Privacy w.r.t. the Particular LSH Function}
\label{sub:privacy:worst}
We first show the privacy guarantee for LSHRR w.r.t. the particular LSH function used by the service provider.
This type of privacy is defined using the Hamming distance $\dV$ between the hash values of given inputs, and 
the degree of privacy depends on the actual selection of the LSH function $H$ (or the hash seeds $\bmr$),
which we assume is available to the adversary.
Since LSH preserves the original metric $\dX$ only approximately, we obtain \XDP{} guarantee w.r.t. a pseudo-metric $\deh$ that approximates $\dX$ as follows.

\begin{restatable}[\XDP{} of $Q_H$]{prop}{worstPrivacyQH}
\label{prop:worstPrivacyQH}
Let $H: \calx\rightarrow\calv$ be a $\kappa$-bit LSH function, and 
$\deh$ be the pseudometric over $\calx$
defined by 
$\deh(\bmx, \bmx') = \varepsilon\dV(H(\bmx), \allowbreak H(\bmx'))$ for $\bmx, \bmx' \in\calx$.
Then the $\varepsilon$-LSHRR $Q_H$ instantiated with $H$
provides $(\deh,0)$-\XDP{}.
\end{restatable}
However, we cannot compute $\deh$ or the degree of \XDP{} in Proposition~\ref{prop:worstPrivacyQH} until $H$ has been computed.
To overcome this unclear guarantee of privacy, in \Sec{sub:expected:privacy:DPLSH} we show a useful privacy guarantee that 
can be evaluated
without requiring 
$H$ (or hash seeds) generated by the service provider.

Note that the $\kappa\varepsilon$-\DP{} of LSHRR is obtained as the worst case of Proposition~\ref{prop:worstPrivacyQH}, i.e., when the hamming distance between vectors is maximum due to 
an ``unlucky'' choice of
hash seeds or 
very large distance
$\dX(\bmx, \bmx')$ between the inputs $\bmx, \bmx'$.
The following proposition guarantees the privacy independently of the actual choice of $H$.

\begin{restatable}[Worst-case privacy of $Q_H$]{prop}{worstDPQH}
\label{prop:worstDPQH}
For a $\kappa$-bit LSH function $H$,
the $\varepsilon$-LSHRR $Q_H$ instantiated with $H$ provides $\kappa \epsilon$-$\DP{}$.
\end{restatable}

\subsection{LSHRR's Privacy w.r.t. the Distribution of LSH Functions}
\label{sub:expected:privacy:DPLSH}

Next, we show LSHRR's privacy guarantee w.r.t. any possible LSH function that may be generated.
This type of privacy guarantee is useful in a variety of scenarios. 
For example, 
a privacy analyst could
evaluate the expected degree of privacy before the service provider fixes the LSH function or hash seeds.
For another example, the seeds 
may be 
stored in tamper-resistant hardware privately.

The privacy guarantee without relying on specific LSH functions or hash seeds is modeled as a probability distribution of degrees of \XDP{} over the random choice of seeds.
Then this can be characterized as an extension of concentrated \DP{}~\cite{Dwork:16:arXiv} and probabilistic \DP{}~\cite{Machanavajjhala:08:ICDE} with input distance,
yielding
the
\XDP{} guarantee.

In the privacy analysis, we deal with the situation where multiple users produce hash values by employing the same hash seeds,
as seen in typical applications such as approximate NNS.
Then we define privacy notions for the mechanisms that share randomness among them.

Formally, we denote 
by $\alg_r: \calx \rightarrow \Dists\caly$ a randomized algorithm $\alg$ with a shared input $r\in\calr$.
Given a distribution $\lambda$ over a finite set $\calr$ of shared input, we denote by $\alg_\lambda: \calx \rightarrow \Dists\caly$ the randomized algorithm that draws a shared input $r$ from $\lambda$ and behaves as $\alg_r$;
i.e., 
$\alg_\lambda(x)[y] = \sum_{r\in\calr} \lambda[r] \alg_r(x)[y]$.
Then we extend the notion of privacy loss~\cite{Dwork:16:arXiv} with shared randomness
as follows.

\begin{definition}[Privacy loss]\label{def:privacy-loss-shared}\rm
For a randomized algorithm $\alg_r: \calx \rightarrow \Dists\caly$ with a shared input $r$,
the \emph{privacy loss} on $y\in\caly$ w.r.t. $x, x'\in\calx$, $r\in\calr$ is defined by:
\begin{align*}
\Loss{x,x',y,r} =
\ln \bigl( {\textstyle\frac{ \alg_r(x)[y] }{ \alg_r(x')[y] }} \bigr)
{,}
\end{align*}
where the probability is taken over the random choices in $\alg_r$.
Given a distribution $\lambda$ over $\calr$,
the \emph{privacy loss random variable} $\Loss{x,x'}$ of $x$ over $x'$ w.r.t. $\lambda$ is the real-valued random variable representing the privacy loss $\Loss{x,x',y,r}$ where a \emph{shared randomness} $r$ is sampled from $\lambda$ and $y$ is sampled from $\alg_r(x)$.
\end{definition}

To characterize the privacy loss random variable  $\Loss{x,x'}$ for 
LSHRR, we introduce an extension of \CDP{}~\cite{Dwork:16:arXiv} wth input distance $d(x, x')$ as follows.
\begin{definition}[\CXDP{}]\label{def:CXDP}\rm
Let $\mu \in\realsnng$, $\tau \in\realspos$, $\lambda\in\Dists\calr$, and $d: \calx\times\calx \rightarrow \realsnng$ be a metric.
A random variable $Z$ over $\reals$ is \emph{$\tau$-subgaussian} if for all $s\in\reals$, $\expect[\exp(s Z)] \leq \exp(\frac{s^2\tau^2}{2})$.
A randomized algorithm $\alg_\lambda: \calx \rightarrow \Dists\caly$ provides \emph{$(\mu, \tau, d)$-mean-concentrated extended differential privacy (\CXDP{})} 
if for all $x, x' \in\calx$, the privacy loss random variable $\Loss{x,x'}$ of $x$ over $x'$ w.r.t. $\lambda$ satisfies that $\expect[\Loss{x,x'}] \le \mu d(x, x')$, and that $\Loss{x,x'} - \expect[\Loss{x,x'}]$ is $\tau$-subgaussian.
\end{definition}

Then we obtain the following \CXDP{} guarantee for 
LSHRR.
\begin{restatable}[\CXDP{} of $\Qlsh$]{prop}{CXDPofMechanism}
\label{prop:CXDPofMechanism}
The $\varepsilon$-LSHRR provides $(\varepsilon \kappa,\frac{\varepsilon \kappa}{2},\dX)$-\CXDP{}.
\end{restatable}

To clarify the implication of \CXDP{}, we introduce an extension of probabilistic \DP{}~\cite{Machanavajjhala:08:ICDE} with input distance, which we call \PXDP{}.
Intuitively, $(\xi, \delta)$-\PXDP{} guarantees $(\xi,0)$-\XDP{} with probability $1-\delta$.

\begin{definition}[\PXDP{}]\label{def:PXDP:general}\rm
Let $\lambda\in\Dists\calr$, $\xi: \calx\times\calx \rightarrow \realsnng$, and $\delta: \calx\times\calx \rightarrow [0, 1]$.
A randomized algorithm $\alg_\lambda: \calx \rightarrow \Dists\caly$ provides \emph{$(\xi, \delta)$-probabilistic extended differential privacy (\PXDP{})} if for all $x, x' \in\calx$, 
$\Prob[\, \Loss{x,x'} > \xi(x, x') \,] \le \delta(x, x')$.
We abuse notation to write $\delta$ when $\delta(x, x')$ is constant.
\end{definition}

\conference{In Appendix~\ref{appendix:properties},} 
\arxiv{In Appendix~\ref{sec:proofs},}
we show that \CXDP{} implies \PXDP{} and that \PXDP{} implies \XDP{}.
Based on these, we show that 
LSHRR provides \PXDP{} and \XDP{} as follows.

\begin{restatable}[\PXDP{}/\XDP{} of $\Qlsh$]{thm}{generalXDPofMechanism}
\label{thm:generalXDPofMechanism}
Let $\delta\in\realspos$, $\varepsilon' = \varepsilon \sqrt{\frac{-\ln\delta}{2}}$, and
$\xi(\bmx, \bmx') = \varepsilon \kappa \dX(\bmx, \bmx') + \varepsilon' \sqrt{\kappa}$.
The $\varepsilon$-LSHRR provides $(\xi, \delta)$-\PXDP{},
hence $(\xi, \delta)$-\XDP{}.
\end{restatable}

For our experimental evaluation, we show a privacy guarantee that gives tighter bounds but requires the parameters dependent on the inputs $\bmx$ and $\bmx'$.

\begin{restatable}[Tighter bound for \PXDP{}/\XDP{}]{prop}{tighterXDPofMechanism}
\label{prop:tighterXDPofMechanism}
For $a, b \in \realspos$, 
let $\DKL(a \| b) = a \ln {\textstyle\frac{a}{b}} + (1-a) \ln {\textstyle\frac{1-a}{1-b}}$.
For an $\alpha\in\realspos$,
we define:
\begin{align*}
\xi_\alpha(\bmx, \bmx') &= \varepsilon \kappa (\dX(\bmx, \bmx') + \alpha) \\
\delta_\alpha(\bmx, \bmx') &= \exp\bigl( - \kappa \DKL(\dX(\bmx, \bmx')+\alpha \| \dX(\bmx, \bmx')) \bigr)
{.}
\end{align*}
Then the $\varepsilon$-LSHRR provides $(\xi_\alpha, \delta_\alpha)$-\PXDP{},
hence $(\xi_\alpha, \delta_\alpha)$-\XDP{}.
\end{restatable}

\subsection{Privacy Guarantee for LapLSH}
\label{sub:privacy-laplsh}

Finally, we also show that LapLSH provides \XDP{}.
This is immediate from the fact that \XDP{} is preserved under the post-processing by an LSH function.

\begin{restatable}[\XDP{} of $\QLapH$ and $\Qlaplsh$]{prop}{generalXDPofLapLSH}
\label{thm:generalXDPofLapLSH}
The $(\varepsilon,\dX)$-LapLSH $\QLapH$ with a $\kappa$-bit LSH function $H$ provides $(\varepsilon \dX,0)$-\XDP{}.
The $(\varepsilon,\dX)$-LapLSH $\Qlaplsh$ w.r.t. a distribution $D_\calh^\kappa$ of the $\kappa$-bit LSH functions also provides $(\varepsilon \dX,0)$-\XDP{}.
\end{restatable}

\section{Experimental Evaluation}
\label{sec:evaluate}
We show an experimental evaluation of LSHRR and LapLSH
on 
two real 
datasets: MovieLens \cite{MovieLens25} and FourSquare \cite{Yang_WWW19}. 
Our goal is to determine the utility of these mechanisms when compared with a (slow but accurate) true nearest neighbor search. 
As a baseline, we 
also
show 
the performance of
vanilla (non-private) LSH.

\subsection{Datasets and Experimental Setup}
\label{exp:setup}\rm

Our problem of interest is \emph{privacy-preserving friend matching (or friend recommendation)}. 
In this scenario, we are given a dataset 
of users in which each user is represented as a (real-valued) vector of attributes. 
The data curator's goal is to 
recommend $k$ friends for each user
based on their $k$-nearest neighbors.

For our experiments, we used the following two datasets:

\smallskip
\noindent{\textbf{MovieLens.}}
The MovieLens 25m dataset \cite{MovieLens25} 
contains $162000$ users with ratings across $62000$ movies, with ratings ranging from $1$ to $5$. 
We normalized the scores, i.e., to mean $0$, and gave unseen movies a score of $0$.
For each user, we 
constructed 
a rating vector
that consists of 
the user's rating for each movie.

\smallskip
\noindent{\textbf{Foursquare.}}
The Foursquare dataset (Global-scale Check-in Dataset with User Social Networks) \cite{Yang_WWW19} contains $90048627$ check-ins by $2733324$ users on POIs all over the world. 
We extracted $107091$ POIs in New York and $10000$ users who have visited at least one POI in New York. 
For each user, we constructed a visit-count vector, which consists of a visit-count value for each POI. 

\smallskip
For both datasets, we generated 
input (rating/visit-count) vectors 
of length 
$n=100, 500, 1000$ 
to evaluate the effectiveness of LSH.
Reduced vector lengths were used because 
LSH has poor utility for 
larger vector lengths
and the utility of our mechanisms requires a good baseline utility for LSH. 

We computed the $k$ nearest neighbors w.r.t. the angular distance $d_\theta$ for 1000 users for $k = 1, 5, 10$ using standard NNS (i.e., pairwise comparisons over all inputs).
The distributions of True Nearest Neighbor distances are shown in \Fig{fig:truenn_distances}. 

\begin{figure}[t]
   \begin{subfigure}[b]{.49\columnwidth}
      \centering
      \caption{MovieLens Dataset}
      \includegraphics[width=.85\linewidth]{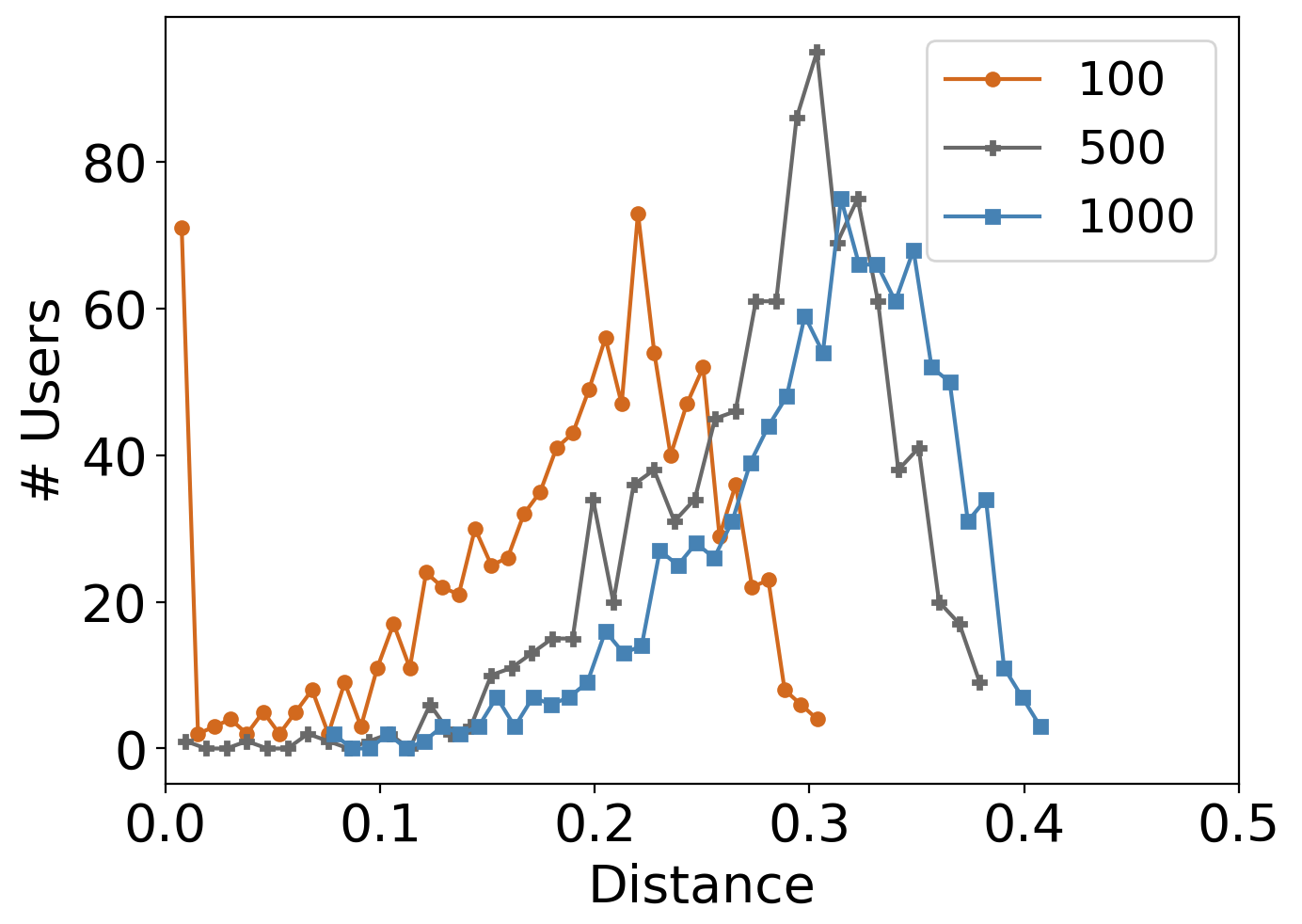}
    \end{subfigure}
    \begin{subfigure}[b]{.49\columnwidth}
      \centering
       \caption{Foursquare Dataset}
      \includegraphics[width=.85\linewidth]{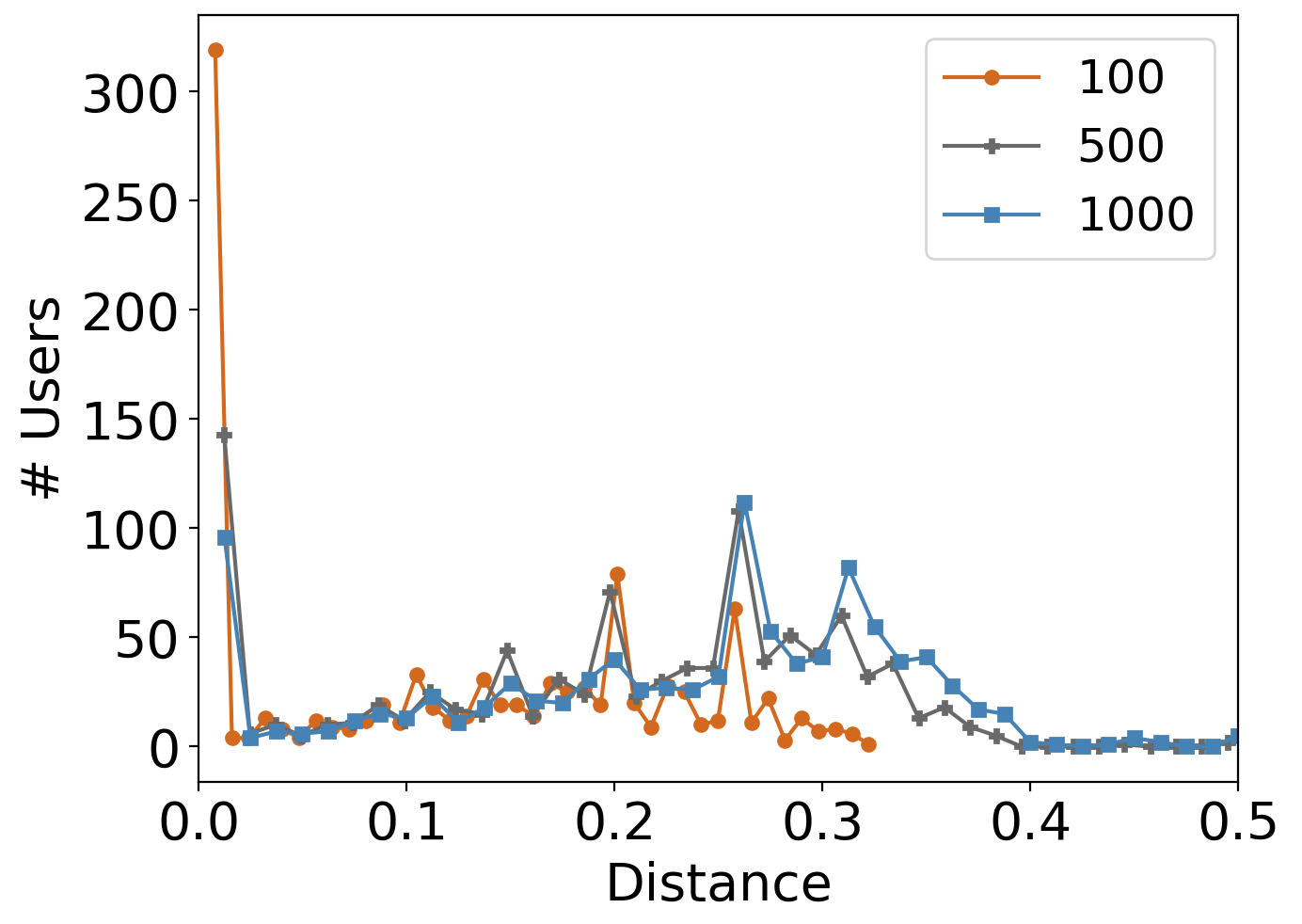}
    \end{subfigure}
  \caption{Distributions of angular distances $d_\theta$ to nearest neighbor for $k = 1$ for each user, plotted for vectors with dimensions 100, 500 and 1000. The distance 0.5 represents orthogonal vectors; i.e., having no items in common. The privacy guarantee for users is a function of the distance $d_\theta$ to their nearest neighbors.}
  \label{fig:truenn_distances}
\end{figure}

$\kappa$-bit LSH (for $\kappa = 10, 20, 50$) was implemented using the \emph{random-projection-based hashing}. 
For each user, we then computed their $k$ nearest neighbors for $k = 1, 5, 10$ using the Hamming distance on bitstrings. %
To compute the overall $(\xi, \delta)$-\XDP{} guarantee as per \Prop{prop:tighterXDPofMechanism}, we fixed $\delta = 0.01$ and $d_\theta = 0.1$ and varied $\varepsilon$ to generate 
$\xi$ 
values in the range $0.1$ to 
$20$.

\Fig{fig:truenn_distances} shows that about $43\%$ (resp.~$16\%$) of input vectors with $100$ (resp.~$1000$) dimensions are within the distance of $0.1$ in the Foursquare dataset. 
Thus, extended \DP{} with $d_\theta = 0.1$ is useful to hide such input vectors.

\subsection{Comparing Privacy and Utility}\label{sect:exp:comparison}
We use the angular distance $d_\theta$ as our utility measure, i.e., to determine similar users for the purposes of recommendations.
For utility loss, we use 
\Def{def:utility-loss}
instantiated with the angular distance $d_\theta$.
We compare the utility loss of each mechanism w.r.t. a comparable privacy guarantee, namely the overall 
privacy budget 
$\varepsilon\deuc(\bmx, \bmx')$ for LapLSH and 
$\xi_\alpha(\bmx, \bmx')$ for LSHRR (\Prop{prop:tighterXDPofMechanism}). 
However, as LSHRR's privacy guarantee depends on the angular distance $d_\theta$ 
and LapLSH's depends on $\deuc$, 
they cannot be compared directly.
For comparison using the same metric,
we 
use 
the relationship between the Euclidean and angular distances
for normalized vectors $\bmx, \bmx'$:
\begin{align}\label{eq:transform:metric}
      \deuc(\bmx, \bmx') ~=~ \sqrt{2 - 2\cos({\pi{\cdot}d_\theta(\bmx, \bmx')})}~.
\end{align}

We normalized input vectors to length 1 (noting that the normalization does not affect the angular distance, hence utility), and transformed 
$\varepsilon\deuc(\bmx, \bmx')$ into $\xi_\alpha(\bmx, \bmx')$ using \eqref{eq:transform:metric}
(Since 
$\xi_\alpha(\bmx, \bmx')$ depends on $\alpha$ and $d_\theta(\bmx, \bmx')$, we perform comparisons against various reasonable ranges of these variables).

We note that the trade-off between privacy and utility means that users with similar profiles will be indistinguishable from each other, whereas users with very different profiles can be distinguished. 
This is an inherent trade-off determined by the correlation between the sensitive and useful information to be released.

\subsection{Experimental Results}

We compared the performance of LapLSH and LSHRR 
with that of vanilla LSH 
in \Fig{fig:utility_lap}. 
We observe that 
LSHRR outperforms LapLSH when the dimension of the input vector is $n=100$, $500$, or $1000$. 
This 
is because LapLSH needs to add noise for each element of the input vector (even if the vector is sparse and includes many zero elements) and the total amount of noise is very large in high-dimensional data. 
In contrast, when the vector length is $n=50$, LapLSH ($\kappa=50$ bits) outperforms LSHRR ($\kappa=50$ bits).
We conjecture that
this is because 
the total amount of noise
used
in LapLSH is small 
for
low-dimensional data
whereas
LSHRR needs to add 
a large amount of
noise for each element of the hash when the hash length $\kappa$ is large. 
We expect LapLSH performance to improve further over LSHRR for smaller values of $n$.

Interestingly, we observe that although the performance of LSH degrades as 
the hash length $\kappa$ 
decreases, the performance of LSHRR and LapLSH both remain relatively stable. 
This is mainly because when 
$\kappa$ is $5$ times larger, the amount of information expressed by the 
hash 
can be roughly $5$ times larger whereas the amount of noise added to each bit is also $5$ times larger.
When the privacy budget is $\xi=20$, 
the performance of LSHRR on larger bit-lengths 
($\kappa=20$ or $50$) 
overtakes the performance of 10-bit LSHRR. This is because the utility 
loss 
of LSHRR is bounded below by the utility loss of the corresponding LSH; i.e., LSHRR converges to LSH with the same hash length $\kappa$ as $\xi$ increases.

\begin{figure}[!ht]
   \begin{subfigure}[b]{.99\columnwidth}
      \centering
       \caption{MovieLens Dataset}
      \includegraphics[width=.99\linewidth]{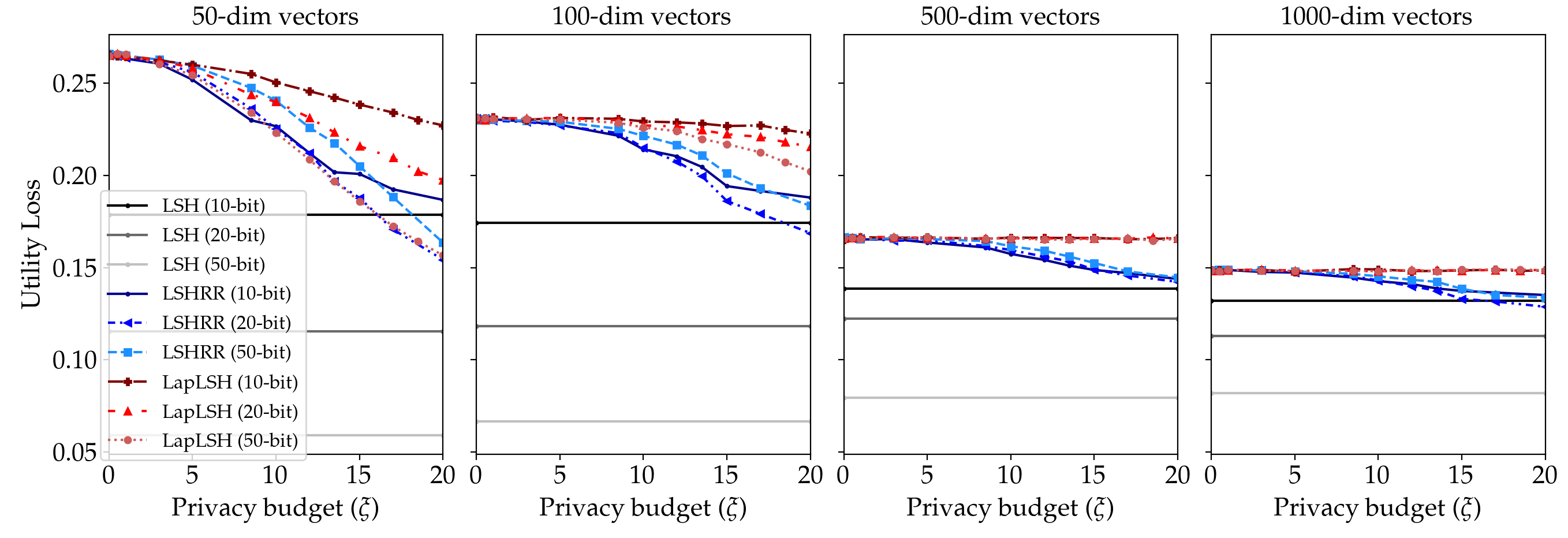}
    \end{subfigure}
 \vspace*{6pt}%
    \begin{subfigure}[b]{.99\columnwidth}
      \centering
      \caption{Foursquare Dataset}
      \includegraphics[width=.99\linewidth]{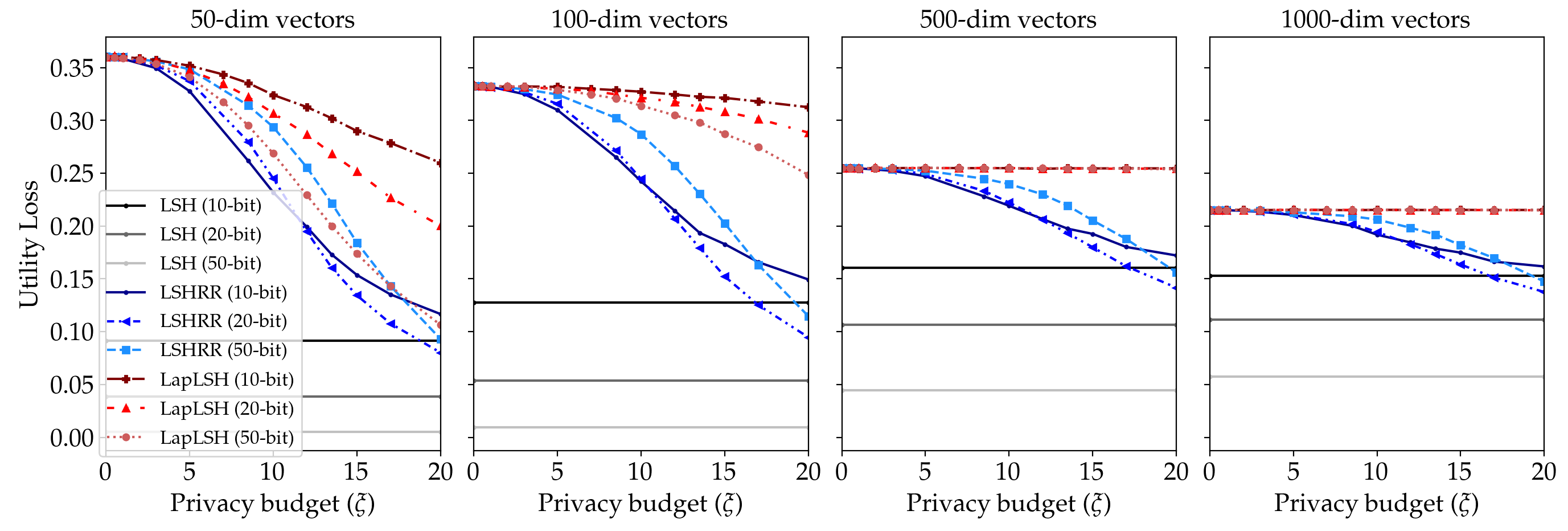}
    \end{subfigure}
      \caption{Utility loss (y-axis) versus privacy budget $\xi$ (x-axis) for LSHRR, LapLSH and LSH on $n$-dimensional vectors. $\xi$ is computed for various
      $\kappa$, and $d_\theta=0.1$.}
  \label{fig:utility_lap}
\end{figure}

\Fig{fig:utility_lap} also shows that when the total privacy budget $\xi$ is around $2$, LSHRR achieves lower utility loss than a uniformly random hash, i.e., 
LSHRR when the total privacy budget is $0$. 
LSHRR achieves 
much lower utility loss 
when the total privacy budget is around $5$. 
We can interpret the value of the total privacy budget in terms of the flip probability in the RR. 
For example, when we use the $20$-bit hash, the total privacy budget of $5$ for $d_\theta = 0.05$ corresponds to the case in which the RR flips each bit of the hash with the probability approx. $0.27$. 
Therefore, we flip around $5$-bits on average out of $20$-bits in this case.

We also note that the total privacy budget 
used in our experiments
is much smaller than the privacy budget $\epsilon$ 
previously used
in the low privacy regime \cite{Kairouz:16:ICML}. 
Specifically, Kairouz \textit{et al.} \cite{Kairouz:16:ICML}, and subsequent works (e.g., \cite{Acharya_AISTATS19,Murakami:19:USENIX,Wang_arXiv16,Ye_ISIT17})
refer to $\epsilon = \ln |\calx|$ as a privacy budget in the low privacy regime. 
Since our experiments deal with high-dimensional data, 
a
privacy budget 
of $\ln |\calx|$ would be
extremely large.
For example, when we use the $1000$-dimensional rating vector in the MovieLens dataset, 
the privacy budget in the low privacy regime is: $\epsilon = \ln |\calx| = \ln 5^{1000} = 1609$. 
The total privacy budget 
in our experiments 
($\xi \leq 20$) 
is much smaller than this value, and falls into the \textit{medium privacy regime} \cite{Acharya_AISTATS19,Ye_ISIT17}.

Note that \LDP{} requires a much larger privacy budget than extended \DP{}. 
For example, 
by Proposition~\ref{prop:tighterXDPofMechanism}, 
when 
$\kappa=50$ and 
$d_\theta=0.05$ (resp.~$0.1$), 
the total privacy budget $\xi=20$ in extended \DP{} corresponds to the total privacy budget 
$120$ (resp.~$80$) in \LDP{}. 
More details 
are shown 
in Appendix~\ref{sec:appendix:experiments}.

Finally, we
compare LSHRR with LapLSH in terms of time complexity and general applicability. 
For time complexity, LapLSH requires 
$O(n\kappa)$ 
operations (construction of $n$-dimensional noise, then $\kappa$-bit hashing). 
In contrast, 
LSHRR requires 
$O(m\kappa)$ operations, 
where $m$ is the number of non-zero elements in the input vector ($\kappa$-bit hashing on 
$m$ 
non-zero elements followed by $\kappa$-randomized response). 
Since $m \ll n$ in practice, 
LSHRR is 
significantly more efficient.

For general applicability, 
LSHRR 
can be used with other metrics such as 
the Jaccard metric \cite{Broder_JCSS00}, 
Earth Mover's metric \cite{Charikar_STOC02}, 
and $l_p$ metric \cite{Datar_SCG04} 
by choosing a suitable 
LSH 
function, 
whereas 
LapLSH 
is designed for the Euclidean metric only. 
Thus, 
LSHRR has more potential applications than LapLSH.

In summary, we find that LSHRR is better than LapLSH in terms of both time complexity and general applicability, and provides high utility with a reasonable privacy level for a high-dimensional data ($100$ dimensions or more).

\subsection{Inapplicability of the RAPPOR}
\label{sub:RAPPOR}
We finally explain that 
neither the RAPPOR \cite{Erlingsson_CCS14} nor the generalized RAPPOR \cite{Wang_USENIX17} 
can be used for
friend matching based on high-dimensional personal data. 
These mechanisms apply a Bloom filter to an input vector before applying the randomized response.
Typically, this Bloom filter is a hash function that neither allows for efficiently finding an input from its hash value, nor preserves the metric $\dX$ over the inputs.
For instance, \cite{Erlingsson_CCS14} uses MD5 to implement the Bloom filter.

Let us consider two approaches to 
perform 
the nearest neighbor search using RAPPOR:
\textit{comparing two hash values} and \textit{comparing two input vectors}. 

In the first approach, the data collector calculates the Hamming distance between obfuscated hash values. 
Then the utility is completely lost, because the Bloom filter does not preserve the metric $\dX$ over the inputs. 
Hence we cannot recommend friends based on the proximity of input vector in this approach.

In the second approach, the data collector tries to invert obfuscated hash values to the original input vector, and calculates the angular distance between the input vectors to find nearest neighbors.
Since the Bloom filter may not allow for efficiently finding an input from its hash value, the data collectors need to perform exhaustive searches, i.e., to compute the hash values of all possible input data $\calx$.
However, this is computationally intractable when the input domain $\calx$ is very large.
In particular, our setting deals with high-dimensional input data (e.g., $|\calx| = 5^{1000}$ in the $1000$-dimensional MovieLens rating vector), and thus it is computationally infeasible to invert hash values into input vectors.

In summary, the first approach (comparing two hashes) results in a complete loss of utility, and the second approach (comparing two input vectors) is computationally infeasible when the input data are in a high-dimensional space. Therefore, the RAPPOR cannot be applied to our problem of friend matching. 
The same issue applies to a generalized version of the RAPPOR \cite{Wang_USENIX17}.

In contrast, our mechanisms can be applied to friend matching even when 
$|\calx|$ is very large, 
because LSH allows us to approximately compare the distance between the input vectors without computing them from their hash values.

\section{Conclusion}
\label{sec:conclusion}
In this paper, we proposed two extended \DP{} mechanisms LSHRR and LapLSH.
We showed that LSH itself does not provide privacy guarantees and could result in complete privacy collapse in some situations.
We then proved that LSHRR and LapLSH provide rigorous guarantees of extended \DP{}.
To our knowledge, this work is the first to provide extended \DP{} with the angular distance.

By experiments with real datasets, we show that 
LSHRR outperforms LapLSH
on high-dimensional data.
We also show that LSHRR provides high utility for a high-dimensional vector,
thus enabling 
friend matching with rigorous privacy guarantees and high utility.

\bibliographystyle{splncs04}
\bibliography{short2}

\appendix

\section{Total Privacy Budgets in Extended \DP{} and \LDP{}}
\label{sec:appendix:experiments}

Table~\ref{tab:lsh_loss} shows total privacy budgets in extended \DP{} and \LDP{} calculated from Proposition~\ref{prop:tighterXDPofMechanism}
and the fact that the angular distance is $0.5$ or smaller.

For example, when $d_\theta=0.05$ 
and $\kappa=10$, $20$, and $50$, 
the total privacy budget $\xi=20$ in extended \DP{} corresponds to the total privacy budget of $55$, $79$, and $120$, respectively, in \LDP{}.

\begin{table}[h]
\centering
\caption{Total privacy budgets in extended \DP{} (\XDP{}) and \LDP{} when $d_\theta = 0.05$ or $0.1$, $\kappa = 10$, $20$, or $50$, and $\delta = 0.01$.}
(a) $d_\theta = 0.05$\\
\begin{tabular}{ c || c | c | c | c }
   \hline
   Total privacy budget $\xi$ in \XDP{} &  1 & 5 & 10 & 20 \\
   \hline
   Total privacy budget in \LDP{} ($\kappa=10/20/50$) &  3/4/6 & 14/20/30 & 28/40/60 & 55/79/120 \\
   \hline
\end{tabular}\\
(b) $d_\theta = 0.1$\\
\begin{tabular}{ c || c | c | c | c }
   \hline
   Total privacy budget $\xi$ in \XDP{} &  1 & 5 & 10 & 20 \\
   \hline
   Total privacy budget in \LDP{} ($\kappa=10/20/50$) &  2/3/4 & 10/14/20 & 21/28/40 & 42/57/80 \\
   \hline
\end{tabular}
\label{tab:lsh_loss}
\end{table}

\arxiv{
\section{Proofs for the Technical Results}
\label{sec:proofs}
We first recall Chernoff bound, which is used in the proof for Lemma~\ref{lem:CXDPimpliesPXDP}.
\begin{restatable}[Chernoff bound]{lem}{Chernoff}
\label{lem:Chernoff}
Let $Z$ be a real-valued random variable.
Then for all $t \in\reals$,
\[
\Pr[ Z \ge t] \leq
\min_{s\in\reals} \frac{\expect[\exp(s Z)]}{\exp(s t)}
{.}
\]
\end{restatable}

Next we recall Hoeffding's lemma, which is used in the proof for Proposition~\ref{prop:CXDPofMechanism}.
\begin{restatable}[Hoeffding]{lem}{Hoeffding}
\label{lem:Hoeffding}
Let $a, b \in \reals$, and $Z$ be a real-valued random variable such that $\expect[Z] = \mu$ and that $a \leq Z \leq b$.
Then for all $t \in\reals$,
\[\expect[ \exp(t Z) ] \le
\exp\bigl(t\mu + {\textstyle\frac{t^2}{8} \bigl( b - a \bigr)^2}\bigr)
{.}
\]
\end{restatable}
Note that Lemma~\ref{lem:Hoeffding} implies that 
$\expect[ \exp(t (Z - \expect[Z])) ] \le
\exp\bigl({\textstyle\frac{t^2}{8} \bigl( b - a \bigr)^2}\bigr)
$.
\vspace{1ex}

Then we recall Chernoff-Hoeffding Theorem, which is used in the proof for Theorem~\ref{thm:generalXDPofMechanism}.
Recall that the Kullback-Leibler divergence $\DKL(a \| b)$ between Bernoulli distributed random variables with parameters $a$ and $b$ is defined by:
\[
\DKL(a \| b) = 
a \ln {\textstyle\frac{a}{b}} + (1-a) \ln {\textstyle\frac{1-a}{1-b}}.
\]

\begin{restatable}[Chernoff-Hoeffding]{lem}{ChernoffHoeffding}
\label{lem:Chernoff-Hoeffding}
Let $Z \sim \Bino(k, p)$ be a binomially distributed random variable where $k$ is the total number of experiments and $p$ is the probability that an experiment yields a successful outcome.
Then for any $\alpha \in\realspos$,
\[
\Pr[ Z \ge k(p + \alpha) ] \le
\exp\bigl( - k \DKL(p+\alpha \| p) \bigr)
{.}
\]
By relaxing this, we have a simpler bound:
\[
\Pr[ Z \ge k(p + \alpha) ] \le
\exp\bigl( - 2 k \alpha^2 \bigr)
{.}
\]
\end{restatable}
\vspace{3ex}

We show the proofs for technical results as follows.

\utilityQH*
\begin{proof}
By the triangle inequality and $Q_H= \Qbrr \circ H$, we have:
\[
\dV(Q_H(x), \allowbreak Q_H(x')) \le \allowbreak \dV(\Qbrr \circ H(x), \allowbreak H(x)) + \dV(H(x), \allowbreak H(x')) + \dV(H(x'), \Qbrr \circ H(x')).\]
It follows from the definition of the bitwise RR $\Qrr$ that for any $\kappa$-bit string $v\in\calv$, the expected Hamming distance is
$\expect[\dV(v, \Qbrr(v))] = \frac{\kappa}{1 + e^{\varepsilon}}$.
Thus $\expect[\dV(\Qbrr \circ H(x), \allowbreak H(x)) + \dV(H(x'), \Qbrr \circ H(x'))] =  \frac{2\kappa}{1 + e^{\varepsilon}}$.
Hence we obtain the proposition.
\myqed
\end{proof}

We present LSHRR's privacy guarantee for hash values, which relies on the $\XDP{}$ of the bitwise RR $\Qbrr$ w.r.t. the Hamming distance $\dV$ as follows.

\begin{restatable}[\XDP{} of BRR]{prop}{XDPofRAPPOR}
\label{prop:XDPofRAPPOR}
The $(\varepsilon, \kappa)$-bitwise RR provides 
$(\epsilon\dV,0)$-$\XDP{}$.
\end{restatable}

\begin{proof}
Recall the definition of the $\varepsilon$-RR $\Qrr$ in Definition~\ref{def:RR}.
Let $r = \frac{1}{e^{\varepsilon}+1}$,
$\bmv = (v_1, v_2, \ldots, v_\kappa) \in\calv$, $\bmv' = (v'_1, v'_2, \ldots, v'_\kappa)\in\calv$, and $\bmy = (y_1, y_2, \ldots, y_\kappa) \in\calv$.
By definition we obtain:
\begin{align*}
\Qbrr(\bmv)[\bmy] &= {\textstyle\prod_{i=1}^{\kappa}\,} r^{|y_i - v_i|} (1-r)^{1 - |y_i - v_i|} \\
\Qbrr(\bmv')[\bmy] &= {\textstyle\prod_{i=1}^{\kappa}\,} r^{|y_i - v'_i|} (1-r)^{1 - |y_i - v'_i|}
{.}
\end{align*}
By $\Qbrr(\bmv')[\bmy] > 0$ and the triangle inequality, we have:
\begin{align*}
{\textstyle \ln\frac{\Qbrr(\bmv)[\bmy]}{\Qbrr(\bmv')[\bmy]}}
&\le
\ln\prod_{i=1}^{\kappa} \bigl( {\textstyle\frac{1-r}{r}} \bigr)^{|v_i - v'_i|}
=
\ln \bigl( {\textstyle\frac{1-r}{r}} \bigr)^{\dV(\bmv, \bmv')}
= \varepsilon \dV(\bmv, \bmv')
{.}
\end{align*}
Therefore $\Qbrr$ provides 
$(\epsilon\dV,0)$-$\XDP{}$.
\myqed
\end{proof}

\worstPrivacyQH*
\begin{proof}
Let $\bmx, \bmx' \in \calx$ and $\bmy\in\calv$.
\begin{align*}
Q_H(\bmx)[\bmy]
&= \Qbrr(H(\bmx))[\bmy]
\\ &\le \varepsilon \dV(H(\bmx), H(\bmx')) \Qbrr(H(\bmx'))[\bmy]
\hspace{-1.5ex}
& \text{(by Proposition~\ref{prop:XDPofRAPPOR})}
\\ &= \varepsilon d_H(\bmx, \bmx') Q_H(\bmx')[\bmy]
& \text{(by the def. of $\deh$)}
\end{align*}
Hence $Q_H$  provides $(\deh,0)$-\XDP{}.
\myqed
\end{proof}

\worstDPQH*
\begin{proof}
Since $d_H(\bmx, \bmx') \le \kappa$ holds for all $\bmx, \bmx'$,
this proposition follows from Proposition~\ref{prop:worstPrivacyQH}.
\end{proof}

\arxiv{
\begin{restatable}[\CXDP{} $\Rightarrow$ \PXDP{}]{lem}{CXDPimpliesPXDP}
\label{lem:CXDPimpliesPXDP}
Let $\mu \in\realsnng$, $\tau \in\realspos$, $\lambda\in\Dists\calr$, $\alg_\lambda: \calx \rightarrow \Dists\caly$, and $d$ be a metric over $\calx$.
Let $\delta \in(0, 1]$,
$\varepsilon = \tau\sqrt{-2\ln\delta}$, and
$\xi(x, x') = \mu d(x, x')+\varepsilon$.
If $\alg_\lambda$ provides $(\mu, \tau, d)$-\CXDP{}, then 
it provides $(\xi, \delta)$-\PXDP{}.
\end{restatable}
}
\conference{\CXDPimpliesPXDP*}

\begin{proof}
\newcommand{\fixtau}[1]{}
Assume that $\alg_\lambda$ provides $(\mu, \tau, d)$-\CXDP{}.
Let $x, x'\in\calx$.
Then we will show 
$\Prob[\, \Loss{x,x'} > \mu d(x, x') + \varepsilon \,] \le \delta$
as follows.

Let $Z = \Loss{x,x'} - \expect[\Loss{x,x'}]$.
By the definition of \CXDP{}, we have:
\begin{align}\label{eq:CXDPimpliesPXDP:mu}
\expect[\Loss{x,x'}] \le \mu d(x, x'), 
\end{align}
and $Z$ is $\tau\fixtau{d(x, x')}$-subgaussian.
Let $t = \tau\sqrt{-2\ln\delta}\fixtau{\cdot d(x, x')}$ and $a = \tau^2$\fixtau{$a = (\tau d(x,x'))^2$}.
\fixtau{By $d(x, x') \neq 0$, we have $a \neq 0$.}
By the definition of subgaussian variables, $\expect[\exp(s Z)] \leq \exp(\frac{a s^2}{2})$ holds for any $s\in\reals$.
Thus we obtain:
\begin{align}
\Pr[ Z \ge t] 
&\leq
\min_{s\in\reals} \frac{\expect[\exp(s Z)]}{\exp(s t)}
 \conference{&\text{(by the Chernoff bound)}}
 \nonumber
\arxiv{\\& & \hspace{-11ex}\text{(by the Chernoff bound in Lemma~\ref{lem:Chernoff}})
 \nonumber}
\\ & \leq
\min_{s\in\reals} \exp\bigl( {\textstyle\frac{a s^2}{2}} - s t \bigr)
& \hspace{-3ex}\bigl(\text{by $\expect[\exp(s Z)] \leq \exp({\textstyle\frac{a s^2}{2}})$}\bigr) \nonumber
\\ & =
\min_{s\in\reals} \exp\bigl( {\textstyle\frac{a}{2} \bigl( s - \frac{t}{a} \bigr)^{\!2}\! - \frac{t^2}{2a}} \bigr)\hspace{-4ex}
\nonumber
\\ & =
\exp\bigl( {\textstyle - \frac{t^2}{2a}} \bigr)
& \bigl(\text{when $s = {\textstyle\frac{t}{a}}$}\bigr)
\label{eq:CXDP:PXDP:Chernoff}
\end{align}
Recall that $\varepsilon = \fixtau{\mu +}\tau\sqrt{-2\ln\delta}$ by definition.
We obtain:
\begin{align*}
& \mathbin{\phantom{=}}
\Pr[ \Loss{x,x'} > \mu d(x, x') +
\varepsilon] 
\\ & \leq
\Pr[ \Loss{x,x'} > \expect[\Loss{x,x'}] + \tau\sqrt{-2\ln\delta}\fixtau{\cdot d(x, x')}] 
& \text{(by \eqref{eq:CXDPimpliesPXDP:mu}
and the def. of $\varepsilon$)}
\\ & =
\Pr[ Z > \tau\sqrt{-2\ln\delta}\fixtau{\cdot d(x, x')}] 
& \text{(by the def. of $Z$)}
\\ & \leq
\exp\bigl( {\textstyle - \frac{( \tau\sqrt{-2\ln\delta}\fixtau{\cdot d(x, x')})^2}{2a}} \bigr)
& \hspace{-10ex}\text{(by \eqref{eq:CXDP:PXDP:Chernoff} and $t = \tau\sqrt{-2\ln\delta}\fixtau{\cdot d(x, x')}$)}
\\ & =
\delta
& \text{(by $a = \tau^2$\fixtau{$a = (\tau d(x,x'))^2$})}
\end{align*}
Therefore the randomized algorithm $\alg_\lambda$ provides 
$(\xi, \delta)$-\PXDP{}.
\myqed
\end{proof}

\arxiv{
\begin{restatable}[\PXDP{} $\Rightarrow$ \XDP{}]
{lem}{PXDPimpliesXDP}
\label{lem:PXDPimpliesXDP}
Let $\lambda\in\Dists\calr$, $\alg_\lambda: \calx \rightarrow \Dists\caly$,
$\xi: \calx\times\calx\rightarrow\realsnng$, and $\delta: \calx\times\calx\rightarrow[0,1]$.
If $\alg_\lambda$ provides $(\xi, \delta)$-\PXDP{}, it provides $(\xi, \delta)$-\XDP{}.
\end{restatable}
}
\conference{\PXDPimpliesXDP*}
\begin{proof}
Assume that $\alg_\lambda$ provides $(\xi, \delta)$-\PXDP{}.
Let $x, x'\in\calx$.
By the definition of $(\xi, \delta)$-\PXDP{}, we have 
$\Prob[\, \Loss{x,x'} > \xi(x, x') \,] \le~\delta(x, x')$.
Let $S \subseteq \caly$.
For each $r\in\calr$, let
$S'_r = \{ y\in S \mid \Loss{x,x',y,r} > \xi(x, x') \}$.
Then $\sum_{r} \lambda[r] \alg_r(x)[S'_r] \leq \delta(x, x')$ and 
for each $r\in\calr$, 
\conference{$\alg_r(x)[S \setminus S'_r] \leq \exp(\xi(x, x')) \cdot \alg_r(x')[S \setminus S'_r]$.}
\arxiv{
\[
\alg_r(x)[S \setminus S'_r] \leq \exp(\xi(x, x')) \cdot \alg_r(x')[S \setminus S'_r].
\]
}
Hence:
\begin{align*}
\alg_\lambda(x)[S]
&= 
\arxiv{{\textstyle \sum_{r}\,} \lambda[r] \alg_r(x)[S]
\\ &=} 
{\textstyle \sum_{r}\,} \lambda[r] \alg_r(x)[S \setminus S'_r] +
{\textstyle \sum_{r}\,} \lambda[r] \alg_r(x)[S'_r]
\\ &\leq
\Bigl( {\textstyle \sum_{r}\,} \lambda[r] \exp(\xi(x, x')) \cdot \alg_r(x')[S \setminus S'_r] \Bigr)
+ \delta(x, x')
\\ &\leq
\exp(\xi(x, x')) \cdot \Bigl( {\textstyle \sum_{r}\,} \lambda[r] \alg_r(x')[S] \Bigr)
+ \delta(x, x')
\\ &\leq
\exp(\xi(x, x')) \cdot \alg_\lambda(x')[S] + \delta(x, x')
{.}
\end{align*}
Therefore $\alg_\lambda$ provides $(\xi, \delta)$-\XDP{}.
\myqed
\end{proof}

To prove the \CXDP{} of 
LSHRR, we show that the Hamming distance between hash values follows a binomial distribution.

\begin{restatable}[Distribution of the Hamming distance of LSH]{lem}{SubGaussianHammingDistance}
\label{lem:Distribution:HammingDistanceOfLSH}
Let $\calh$ be an LSH scheme w.r.t. a metric $\dX$ over $\calx$ coupled with a distribution $D_\calh$.
Let $\bmx, \bmx' \in\calx$ be any two inputs, and
$Z$ be the random variable of the Hamming distance between their $\kappa$-bit hash values,
i.e., $Z = \dV(H(\bmx), H(\bmx'))$
where a $\kappa$-bit LSH function $H$ is drawn from the distribution $D_\calh^\kappa$.
Then $Z$ follows the binomial distribution with mean $\kappa \dX(\bmx, \bmx')$ and variance $\kappa \dX(\bmx, \bmx')(1 - \dX(\bmx, \bmx'))$.
\end{restatable}

\begin{proof}
By the definition of the Hamming distance $\dV$ and the construction of the LSH-based $\kappa$-bit function $H$, we have
$\dV(H(\bmx), \allowbreak H(\bmx'))
= \sum_{i=1}^{\kappa} |\, h_i(\bmx) - h_i(\bmx') \,|$.
Since $\sum_{i=1}^{\kappa} |\, h_i(\bmx) - h_i(\bmx') \,|$ represents the number of non-collisions between hash values of $\bmx$ and $\bmx'$,
it follows the binomial distribution 
with mean $\kappa \dX(\bmx, \bmx')$ and variance $\kappa \dX(\bmx, \bmx')(1 - \dX(\bmx, \bmx'))$.
\myqed
\end{proof}

\CXDPofMechanism*
\begin{proof}
For a $\kappa$-bit LSH function $H \in \calh^\kappa$, 
\begin{align*}
Q_H(\bmx)[y] 
&=
\Qbrr(H(\bmx))[y]
\\ &\le 
e^{\varepsilon\dV(H(\bmx), H(\bmx'))} \Qbrr(H(\bmx'))[y]
\hspace{3ex} \text{(by Proposition~\ref{prop:XDPofRAPPOR})}
\\ &=
e^{\varepsilon\dV(H(\bmx), H(\bmx'))} Q_H(\bmx')[y]
{.}
\end{align*}
Let $Z$ be the random variable defined by
$Z \eqdef \dV(H(\bmx), H(\bmx'))$
where $H = (h_1, h_2, \ldots, h_\kappa)$ is distributed over $\calh^\kappa$, namely, the seeds of these LSH functions are chosen randomly.
Then $0 \leq Z \leq \kappa$.
By Lemma~\ref{lem:Distribution:HammingDistanceOfLSH},
$Z$ follows the binomial distribution with mean $\expect[Z] = \kappa \dX(\bmx, \bmx')$.
Then the random variable $\varepsilon Z - \expect[\varepsilon Z]$ is centered, i.e., $\expect[\varepsilon Z - \expect[\varepsilon Z]] = 0$, and ranges over $[-\varepsilon \kappa \dX(\bmx, \bmx'), \varepsilon \kappa (1 - \dX(\bmx, \bmx'))]$.
Hence it follows from Hoeffding's lemma
\arxiv{(Lemma~\ref{lem:Hoeffding})}
that: 
\[\expect[ \exp(t(\varepsilon Z - \expect[\varepsilon Z])) ] \le
\exp\bigl({\textstyle\frac{t^2}{8} \bigl(\varepsilon \kappa \bigr)^2}\bigr)
=
\exp\bigl({\textstyle\frac{t^2}{2} \bigl(\frac{\varepsilon \kappa}{2} \bigr)^2}\bigr)
{.}
\]
Hence by definition, $\varepsilon Z - \expect[\varepsilon Z]$ is $\frac{\varepsilon \kappa}{2}$-subgaussian.
Therefore, the LSH-based mechanism $\Qlsh$ provides $(\varepsilon \kappa,\frac{\varepsilon \kappa}{2},\dX)$-\CXDP{}.
\myqed
\end{proof}

\generalXDPofMechanism*
\begin{proof}
Let $\alpha = \sqrt{\frac{-\ln\delta}{2\kappa}}$.
Let $Z$ be the random variable defined by
$Z \eqdef \dV(H(\bmx), \allowbreak H(\bmx'))$
where $H = (h_1, h_2, \ldots, h_\kappa)$ is distributed over $\calh^\kappa$.
By Lemma~\ref{lem:Distribution:HammingDistanceOfLSH},
$Z$ follows the binomial distribution with mean $\expect[Z] = \kappa \dX(\bmx, \bmx')$.
\arxiv{Hence it follows from Chernoff-Hoeffding theorem (Lemma~\ref{lem:Chernoff-Hoeffding}) that:
\[
\Pr[ Z \ge \kappa(\dX(\bmx, \bmx') + \alpha) ] \le
\exp\bigl( - 2 \kappa \alpha^2 \bigr) = \delta
{.}
\]
}
\conference{
Hence it follows from Chernoff-Hoeffding theorem that
$\Pr[ Z \ge \kappa(\dX(\bmx, \bmx') + \alpha) ] \le
\exp\bigl( - 2 \kappa \alpha^2 \bigr) = \delta$.}
Hence 
$\Pr[ \varepsilon Z \ge \varepsilon \kappa \dX(\bmx, \bmx') + \varepsilon' \sqrt{\kappa} ] \le \delta$.
Therefore $\Qlsh$ provides $(\xi, \delta)$-\PXDP{}.
By Lemma~\ref{lem:PXDPimpliesXDP}, $\Qlsh$ provides $(\xi, \delta)$-\XDP{}.
\myqed
\end{proof}

\tighterXDPofMechanism*
\begin{proof}
Let $Z$ be the random variable defined by
$Z \eqdef \dV(H(\bmx), \allowbreak H(\bmx'))$
where $H = (h_1, h_2, \ldots, h_\kappa)$ is distributed over $\calh^\kappa$.
\conference{By Chernoff-Hoeffding theorem,
$\Pr[ Z \ge \kappa(\dX(\bmx, \bmx') + \alpha) ] \le \delta_\alpha(\bmx, \bmx')$.
}
\arxiv{By Chernoff-Hoeffding theorem (Lemma~\ref{lem:Chernoff-Hoeffding}),
\[
\Pr[ Z \ge \kappa(\dX(\bmx, \bmx') + \alpha) ] \le \delta_\alpha(\bmx, \bmx')
{.}
\]
}
Then 
$\Pr[ \varepsilon Z \ge \xi_\alpha(\bmx, \bmx') ] \le \delta_\alpha(\bmx, \bmx')$.
Therefore $\Qlsh$ provides $(\xi_\alpha, \delta_\alpha)$-\PXDP{}.
By Lemma~\ref{lem:PXDPimpliesXDP}, $\Qlsh$ provides $(\xi_\alpha, \delta_\alpha)$-\XDP{}.
\myqed
\end{proof}

\generalXDPofLapLSH*
\begin{proof}
Since the application of an LSH function is post-processing, the proposition follows from the \XDP{} of the Laplace mechanism.
\myqed
\end{proof}

}

\end{document}